\documentclass[10pt,article]{IEEEtran}

\usepackage{epsfig,latexsym,amsmath,amssymb,setspace,cite,color,graphicx,algorithm,algorithmic,cases, subcaption}

\usepackage{physics}
\newcommand\numberthis{\addtocounter{equation}{1}\tag{\theequation}}
\makeatletter

\newcommand{\pushright}[1]{\ifmeasuring@#1\else\omit\hfill$\displaystyle#1$\fi\ignorespaces}
\makeatother

\usepackage{amsmath}
\usepackage{amsfonts}
\usepackage{amssymb}
\usepackage{dsfont}
\usepackage{bm}
\usepackage{amsthm}
\usepackage{newlfont}
\usepackage{float}
\usepackage{hyperref}
\usepackage{algorithm}
\usepackage{algorithmic}
\usepackage{enumerate}
\usepackage{chngcntr}
\usepackage{mathtools}
\usepackage{breqn}
\usepackage{stmaryrd}
\usepackage{cite}
\hypersetup{
    colorlinks=true,  
    linkcolor= blue,  
    citecolor=blue,  
        urlcolor=black  
}

\title{Private Sum Computation: Trade-Offs between Communication, Randomness, and Privacy}

\begin{document}
\sloppy
\allowdisplaybreaks[1]

\newtheorem{thm}{Theorem} 
\newtheorem{lem}{Lemma}
\newtheorem{prop}{Proposition}
\newtheorem{property}{Property}
\newtheorem{cor}{Corollary}
\newtheorem{defn}{Definition}
\newcommand{\remarkend}{\IEEEQEDopen}
\newtheorem{remark}{Remark}
\newtheorem{rem}{Remark}
\newtheorem{ex}{Example}
\newtheorem{pro}{Property}

\renewcommand{\qedsymbol}{ \begin{tiny}$\blacksquare$ \end{tiny} }

\renewcommand{\algorithmicrequire}{\textbf{Input:}}
\renewcommand{\algorithmicensure}{\textbf{Inputs:}}

\renewcommand{\leq}{\leqslant}
\renewcommand{\geq}{\geqslant}

\author{R\'emi A. Chou,  J\"org Kliewer, Aylin Yener
\thanks{R\'{e}mi A. Chou is with the Department of Computer Science and Engineering, The University of Texas at Arlington, Arlington, TX 76019. J. Kliewer is with the Department of Electrical and Computer
Engineering, New Jersey Institute of Technology, Newark, NJ. Aylin Yener is with the Department of Electrical and Computer Engineering, The Ohio State University, Columbus, OH 43210. E-mails: remi.chou@uta.edu, jkliewer@njit.edu, yener@ece.osu.edu. A preliminary version of this work was presented at the 2024 IEEE International Symposium on Information Theory (ISIT) in \cite{chou2024}. This work was supported in part by NSF grants CCF-2201824 and CCF-2425371, and in part by Chope Chair funds.}}

\IEEEoverridecommandlockouts

\maketitle
\begin{abstract}

Consider multiple users and a fusion center. Each user possesses a  sequence of bits and can communicate with the fusion center through a one-way public channel. The fusion center's task is to compute the sum of all the sequences under the privacy requirement that a set of colluding users, along with the fusion center, cannot gain more than a predetermined amount $\delta$ of information, measured through mutual information, about the sequences of other users. 
Our first contribution is to characterize the minimum amount of necessary communication between the users and the fusion center, as well as the minimum amount of necessary  randomness at the users. Our  second contribution is to establish a connection between private sum computation and secret sharing by showing that secret sharing is necessary to generate the local randomness needed for private sum computation, and prove that it holds true for any $\delta \geq 0$.
\end{abstract}

\section{Introduction}

The study of distributed summation under security constraints is closely related to the problem of secure aggregation, as evidenced by prior works including \cite{bonawitz2017practical, bell2020secure, so2021turbo, kadhe2020fastsecagg, zhao2022information, so2022lightsecagg, jahani2023swiftagg+, wan2022information, schlegel2023codedpaddedfl, liu2022efficient}. This research area finds applications in distributed computing. Notably, secure summation with zero information leakage has been investigated extensively,  with optimal communication and randomness rates analyzed in \cite{chor1993communication, hayashi2018secure, zhao2021expand, wan2022secure, zhao2022information}.

In this paper, we expand the scope of secure sum computation to allow a controlled amount of information leakage and establish the trade-offs between information leakage, communication rates, and required randomness. In particular, we aim to quantify the communication and randomness gains achievable at the cost of a controlled amount of privacy leakage, which is particularly relevant in settings where communication bandwidth is limited or the generation of randomness is a costly resource to implement. More specifically, we consider a scenario with $L \geq 2$ users, each possessing a private sequence, and a fusion center tasked with computing the sum of these sequences from an encoded version of each sequence sent by each user.  To this end, the users utilize shared randomness, referred to as global randomness in this paper, whose rate must be minimized. From this global randomness is extracted local randomness, which each user uses to encode their sequence before sending it to the fusion center.
The privacy requirement is to ensure that any colluding set of users of size $T$, where $T \leq L-2$ is fixed, along with the fusion center, must not learn more than a predetermined amount $\delta$ of information about the other users' sequences. This information leakage is measured through the conditional mutual information between  all the encoded sequences and all the non-encoded sequences given all the knowledge available at the colluding users and the fusion center. Furthermore, inspired by ramp secret sharing, e.g.,~\cite{blakley1984security, yamamoto1986secret}, we  introduce a private summation setting where the controlled amount of information leakage treats all users symmetrically, ensuring that no user has privileged access to others’ private data and that privacy guarantees are consistent across all users and collusion scenarios.

Our first contribution is deriving capacity results on the amount of communication and randomness needed. Specifically, we establish converse results on the necessary communication rate of individual users to the fusion center and the required rate of shared randomness among users. These results are primarily established through combinatorial arguments. We also derive converse results for the individual rates of local randomness needed by each user when a symmetry assumption holds on the information leakage. Additionally, we provide an achievability scheme that simultaneously matches all the aforementioned converse bounds and thus establishes a capacity result. Our second contribution establishes a fundamental connection between secret sharing and private summation. Specifically, we demonstrate that achieving optimal communication and randomness rates requires that the local randomness of each user corresponds to the shares in a ramp secret sharing scheme. We prove that this condition is both necessary and sufficient for optimality.

\emph{Related work}: The closest related work to this paper  is \cite{zhao2022secure}, which addresses a special case of our setting where (i) no information leakage is allowed, i.e., $\delta =0$, (ii) the individual communication rates of the users are all equal, and (iii) the individual local randomness rates of the users are all equal. The capacity region for this special case has been established in \cite[Theorem 1]{zhao2022secure}. When $\delta=0$, meaning no information leakage is allowed, our results extend \cite[Theorem 1]{zhao2022secure} by providing a capacity result for the communication rates of individual users,  the  local randomness rates of individual users, and the local randomness sum rate  of all the users, without assuming that (ii) and (iii) hold. Additionally, our results extend \cite[Theorem 1]{zhao2022secure} to the case $\delta>0$, and establish capacity results for the communication rates of individual users, the local randomness sum rate  of all the users, and the global randomness rate,  without assuming that (ii) and (iii) hold. Our results also establish a capacity result for the  local randomness rates of individual users when $\delta>0$ under some information leakage symmetry condition.  The achievability part of \cite[Theorem 1]{zhao2022secure} also establishes, when $\delta =0$,  that secret sharing can be employed to generate the local randomness at the users for secure summation. In this study, we establish a stronger connection between secure summation and secret sharing by showing that secret sharing is not only sufficient but also \emph{necessary} to generate the local randomness needed for private summation to achieve optimality. Additionally, beyond the case $\delta=0$, we prove that it holds true for any $\delta \neq 0$.

We note that other studies have explored secure computation with interactive communication from arbitrarily correlated randomness in various settings, e.g., \cite{tyagi2013distributed, data2020interactive, data2016communication, tyagi2011function}. Additionally,  computation from correlated randomness that allows information leakage  is also studied in \cite{tu2019function, chou2022function}. The main difference between the above studies and this work is that the structure of the shared randomness available at the users is part of the coding scheme design in our setting, similar to \cite{zhao2022secure}, whereas in  \cite{tyagi2013distributed, data2020interactive, data2016communication, tyagi2011function,tu2019function, chou2022function} the statistics of the shared randomness available at the users are predetermined and not subject to~design.

The remainder of the paper is organized as follows. The problem statement is presented in Section \ref{secps}. Our main results are summarized in Section \ref{secmr}. The proofs of our converse results are presented in Sections \ref{seccv} and \ref{secthi}. The necessity of secret sharing to generate local randomness is presented in Section~\ref{appproofth4}. Finally, concluding remarks are presented in Section \ref{sec:cl}.

\section{Problem Statement} \label{secps}
Notation: Let $\mathbb{N}$, $\mathbb{R}$, and $\mathbb{Q}$ be the sets of natural, real, and rational numbers, respectively. For $a,b \in \mathbb{R}$, define $\llbracket a,b \rrbracket \triangleq [\lfloor a \rfloor , \lceil b \rceil ] \cap \mathbb{N}$, and $[a] \triangleq \llbracket 1, a\rrbracket$.  The indicator function is denoted by $\mathds{1}\{ \omega \}$, which is equal to~$1$ if the predicate $\omega$ is true and $0$ otherwise.

Consider a fusion center and $L$ users who have individual private sequences. The users  communicate with the fusion center over a one-way, public, and noiseless channel, with the aim  that the fusion center computes the sum of their sequences as described next.

\begin{defn} \label{def1}
An $(L,n,(R^{(X)}_l)_{l\in [L]},R^{(U)},(R^{(K)}_l)_{l\in [L]})$ private-sum computation protocol consists of 
\begin{itemize}
    \item $L$ users indexed in the set $[L]$;
    \item $L$ independent sequences $(S_l)_{l \in [L]}$, where Sequence $S_l$ is owned by User $l\in [L]$ and is uniformly distributed over the finite field $\mathbb{F}_2^n$;
    \item A source of global randomness, independent of the sequences $(S_l)_{l \in [L]}$, and described by $U$, uniformly distributed over $\mathbb{F}_2^{nR^{(U)}}$;
    \item $L$ encoding functions $e_l : \mathbb{F}_2^{nR^{(U)}} \to \mathbb{F}_2^{nR^{(K)}_l} $;
    \item $L$ encoding functions $e^{(X)}_l : \mathbb{F}_2^{nR^{(K)}_l} \times \mathbb{F}_2^{n} \to \mathbb{F}_2^{nR^{(X)}_l}$;
    \item A decoding function $d: \bigtimes_{l\in [L]} \mathbb{F}_2^{nR^{(X)}_l} \to \mathbb{F}_2^n$;
\end{itemize}
and operates as follows:
\begin{enumerate}
    \item User $l\in[L]$ receives the local randomness $K_l \triangleq e_l (U)$;
    \item User $l\in[L]$ encodes the sequence $S_l$ as $X_l \triangleq e^{(X)}_l(K_l,S_l)$ and sends $X_l$ over the public channel;
    \item The fusion center computes $\hat{\Sigma}_{[L]} \triangleq d(X_{[L]})$ an estimate of $\Sigma_{[L]} \triangleq \sum_{l\in[L]} S_l$, where   $X_{[L]} \triangleq (X_l)_{l\in [L]}$.
\end{enumerate}
\end{defn}
Then, we define the desired requirements for an $(L,n,(R^{(X)}_l)_{l\in [L]},R^{(U)},(R^{(K)}_l)_{l\in [L]})$ private-sum computation protocol as in Definition \ref{def1} as follows.
\begin{defn} \label{def2}
Let $\delta \geq 0$ and $T \in \llbracket 0,L-2 \rrbracket$. An $(L,n,(R^{(X)}_l)_{l\in [L]},R^{(U)},(R^{(K)}_l)_{l\in [L]})$ private-sum computation protocol is $(\delta,T)$-private if
    \begin{align}
   \hat{\Sigma}_{[L]}  & =  \Sigma_{[L]} , \label{eqrec} \\
      \max_{\substack{\mathcal{T} \subset [L]\\:  |\mathcal{T}| \leq T}}  \frac{I(S_{[L]}; X_{[L]} |\Sigma_{[L]}, S_{\mathcal{T}} ,K_{\mathcal{T}})}{n} &\leq   \delta,\label{eqprivact}
\end{align}
    where for any $\mathcal{T} \subset [L]$,  $S_{\mathcal{T}} \triangleq (S_l)_{l\in \mathcal{T}}$ and $K_{\mathcal{T}} \triangleq (K_l)_{l\in \mathcal{T}}$.
\end{defn}
 Equation \eqref{eqrec} means that the fusion center computes the sum $\Sigma_{[L]}$ without errors.  If the fusion center and a set $\mathcal{T}$ of $T$ users collude,  then, Equation \eqref{eqprivact} quantifies the normalized amount of information that $X_{\mathcal{T}^c}$ leaks about  $S_{\mathcal{T}^c}$ given the knowledge of $(\Sigma_{[L]}, S_{\mathcal{T}}, K_{\mathcal{T}})$, indeed, note that $I(S_{[L]}; X_{[L]} |\Sigma_{[L]}, S_{\mathcal{T}} ,K_{\mathcal{T}}) = I(S_{\mathcal{T}^c}; X_{\mathcal{T}^c} |\Sigma_{[L]}, S_{\mathcal{T}}, K_{\mathcal{T}})$ by Definition~\ref{def1}. The setting is depicted in Figure~\ref{fig}.
\begin{figure}[ht]
\centering
 \includegraphics[width=7.5cm]{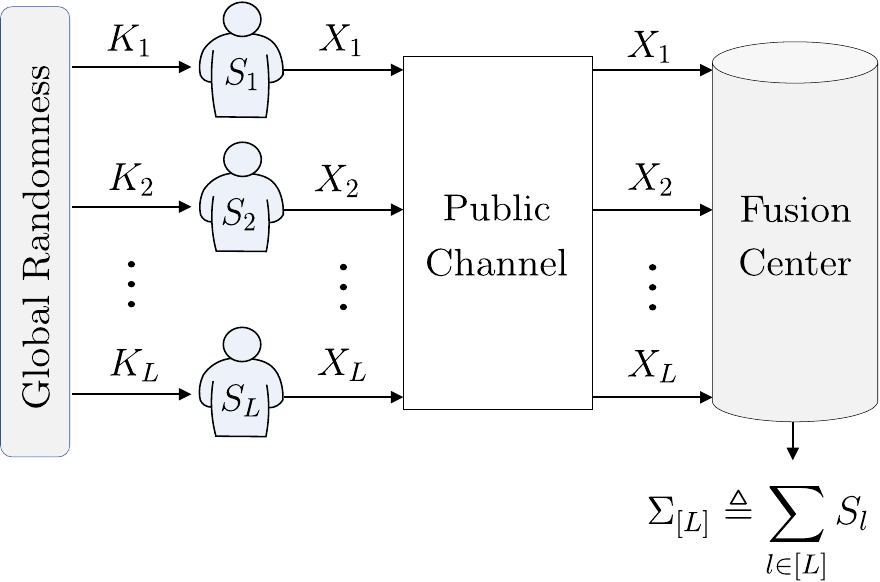}
 \caption{Sum computation setting: $(S_i)_{i\in[L]}$ and $(X_i)_{i\in[L]}$  represent the the private sequences at the users and their encoded versions, respectively. $(K_i)_{i\in[L]}$   represents  local randomness.}
   \label{fig}
\end{figure}

In this study, for protocols that satisfy the requirements of Definition \ref{def2}, we are interested in characterizing (i) the minimum amount of global randomness $U$ needed, (ii) the minimum amount of public communication needed for each user, (iii) the minimum amount of local randomness needed for the users. To this end, we introduce the following quantities.
\begin{defn}
    Let $\delta \geq 0$ and $T \in \llbracket 0,K-2\rrbracket$. Let $\mathcal{C}$ be the set of tuples $\Lambda \triangleq (R^{(X)}_l)_{l\in [L]},R^{(U)},(R^{(K)}_l)_{l\in [L]})$ such that there exist $(L,n,(R^{(X)}_l)_{l\in [L]},R^{(U)},(R^{(K)}_l)_{l\in [L]})$ private-sum computation protocols that are $(\delta,T)$-private. Then, define
    \begin{align*}
        R_{l,\star}^{(X)} &\triangleq \inf_{\Lambda \in \mathcal{C}}  R_{l}^{(X)} , \forall l\in[L],\\
        R^{(K)}_{l,\star} &\triangleq \inf_{\Lambda \in \mathcal{C}}  R^{(K)}_{l} , \forall l\in[L],\\
    R^{(K)}_{\Sigma,\star} &\triangleq \inf_{\Lambda \in \mathcal{C}} \sum_{l \in [L]} R^{(K)}_l, \\
            R_{\star}^{(U)} &\triangleq \inf_{\Lambda \in \mathcal{C}}  R^{(U)}.    \end{align*}
\end{defn}

\begin{ex}[Setting in \cite{zhao2022secure}]\label{ex1}
Consider the special case $\delta = 0$.  In that setting, the following achievability scheme scheme is described in \cite{zhao2022secure}.  The global randomness is $U \triangleq (N_l)_{l\in[L-1]}$, where the $N_l$'s are independent and uniformly distributed over $\mathbb{F}_2^n$. $U$ is then encoded as 
\begin{align*}
    K_l &= e_l (U) \triangleq N_l, \text{ for } l\in [L-1], \\
    K_L&= e_L (U) \triangleq - \sum_{l\in[L-1]} N_l.
\end{align*}
After receiving $X_l$, $l\in [L] $, User $l$ encodes the sequence $S_l$ as 
\begin{align*}
X_l = e_l^{(X)}(K_l,S_l) \triangleq K_l + S_l, \text{ for } l\in [L].
\end{align*}
 Finally, the fusion center recovers $\Sigma_{[L]}$ as \begin{align*} d(X_{[L]}) \triangleq \sum_{l\in[L]} X_l = \sum_{l\in[L]} S_l = \Sigma_{[L]},\end{align*} and the privacy constraint~\eqref{eqprivact} is satisfied by the one-time pad lemma, e.g., \cite[Th. 2.4]{stinson2005cryptography}. Hence, for $l\in [L]$, the rates $R_l^{(X)}=1$, $R_l^{(K)}=1$, and  $R^{(U)}=L-1$ are achieved.

Additionally, \cite{zhao2022secure} proves the optimality of this achievability scheme under the assumptions that (i) the local randomness rates are identical for all users (i.e.,  $\exists C_1 \in \mathbb{R}, R^{(K)}_l = C_1 , \forall l \in [L]$), and (ii) the communication rates are identical for all users (i.e.,  $\exists C_2 \in \mathbb{R}, R^{(X)}_l = C_2 , \forall l \in [L]$). 
\end{ex}

\begin{rem}
In Example \ref{ex1}, the generation of local randomness from global randomness corresponds to a secret sharing scheme, as detailed in Section \ref{secsc}. Hence, Example \ref{ex1} demonstrates the sufficiency of secret sharing for creating local randomness from global randomness. Furthermore, in Section \ref{secsc}, we will prove that secret sharing is not only sufficient but also necessary for this purpose.
\end{rem}

\begin{rem}
   When $\delta = 0$, our model recovers the setting presented in \cite{zhao2022secure}, as described in Example \ref{ex1}. However, unlike \cite{zhao2022secure}, for $\delta = 0$, our setting does not rely on assumptions (i) and (ii), described in Example \ref{ex1}, to establish the converse results.
\end{rem}

\begin{rem}
In Definition \ref{def2}, it is sufficient to consider $T \in \llbracket0,L-2\rrbracket$. If $T=L$, then \eqref{eqprivact} is always satisfied as the left-hand side is equal to zero. This is also true if $T=L-1$, as $S_{[L]}$ can be reconstructed from $(\Sigma_{[L]} ,S_{\mathcal{T}})$ for any $\mathcal{T} \subset [L]$ such that  $|\mathcal{T}| = L-1$.
\end{rem}

\section{Main results} \label{secmr}
We first show in Section \ref{secpre} that Definition \ref{def2} can be simplified without loss of generality. We then present our converse and capacity results in Sections \ref{secconverse} and \ref{secca}, respectively. Finally, we establish the connection between secret sharing and required randomness in private summation in Section \ref{secsc}.
\subsection{Preliminaries} \label{secpre}
As proved in Appendix \ref{app1}, in the privacy constraint \eqref{eqprivact} of Definition \ref{def2}, it is sufficient to consider $\delta$ of the form $\delta =  \alpha (L-1)$ with $\alpha \in [0,1]$, i.e., \eqref{eqprivact} can be replaced by
    \begin{align}
         \max_{\substack{\mathcal{T} \subset [L]\\:  |\mathcal{T}| \leq T}}  I(S_{[L]}; X_{[L]} |\Sigma_{[L]} ,S_{\mathcal{T}} ,K_{\mathcal{T}}) & \leq  n \alpha (L-1).\label{eqprivact2}
\end{align}

\subsection{Converse results} \label{secconverse}
We first derive a converse on the individual communication rates, the local randomness sum-rate, and the global randomness rate.
\begin{thm}[Converse]\label{thmcv}
Consider any $(L,n, (R^{(X)}_l)_{l\in [L]},$ $R^{(U)},(R^{(K)}_l)_{l\in [L]})$ private-sum computation protocol that is $(\delta= \alpha (L-1),T)$-private, where $\alpha \in [0,1]$. Then, we have
\begin{align}
      R^{(X)}_l & \geq 1,  \forall l \in [L], \label{eqth21} \\
    \textstyle\sum_{l \in [L]} R^{(K)}_l & \geq (1-\alpha) L,\label{eqth22}\\
    R^{(U)} & \geq (1-\alpha) (L-1).\label{eqth23}
\end{align}
\end{thm}
\begin{proof}
    See Section \ref{seccv}.
\end{proof}
We then derive a converse on the individual rate of local randomness under the following leakage symmetry assumption:

\begin{enumerate}[(*)]
        \item \emph{For any set of colluding users $\mathcal{T} \subseteq [L-2]$, adding a single user to $ \mathcal{T} $ alters the leakage about the inputs of the remaining users  by a fixed amount, independent of the choice of $ \mathcal{T} $ or the user added.}
    \end{enumerate} 
This assumption  reflects a model that treats all users symmetrically, ensuring that no user has privileged access to others’ private data and that privacy guarantees are consistent across all users and collusion scenarios. The fixed change in leakage implies that each additional user contributes the same amount to the leakage, regardless of the subset they join or their identity.
For $ \mathcal{T} \subseteq [L]$, define the leakage $$\mathbb{L}(\mathcal{T}) \triangleq I(S_{[L]}; X_{[L]} |\Sigma_{[L]} S_{\mathcal{T}} K_{\mathcal{T}}),$$ then the leakage symmetry axiom (*) is mathematically expressed as follows:
\begin{align}
\exists \Delta \in \mathbb{R},    \forall \mathcal{T} \subseteq [L-2], \forall j \in \mathcal{T}^c, \mathbb{L}(\mathcal{T}) - \mathbb{L}(\mathcal{T}\cup\{j\}) = \Delta. \label{eqLsym}
\end{align}
From \eqref{eqLsym}, we deduce the following properties, whose proofs can be found in Appendix \ref{App_properties}.
    \begin{property} \label{prop1}We have
    \begin{align*}
    \forall l \in \llbracket 0,L\rrbracket, \exists C_l \in \mathbb{R}_+, \forall \mathcal{T} \subseteq [L],  |\mathcal{T}|=l\implies \mathbb{L}(\mathcal{T})= C_l .
\end{align*}
\end{property}
\begin{property}\label{prop2}Consider $(C_l)_{l \in \llbracket 0, L\rrbracket}$ as in Property \ref{prop1}. Then, we have\begin{align*}
    \forall l \in \llbracket 0,L-1\rrbracket, C_l \geq C_{l+1}.
\end{align*}
\end{property}
 \begin{property}\label{prop3} Consider $(C_l)_{l \in \llbracket 0, L\rrbracket}$ as in Property \ref{prop1}. Then, we have \begin{align*}
        \Delta & \leq n\alpha ,\\
        C_0 & \leq n\alpha (L-1).
\end{align*}
\end{property}

\begin{thm}[Converse]\label{thmi}
Suppose that leakage symmetry, i.e., \eqref{eqLsym}, is required. Then, for any $(L,n,$ $(R^{(X)}_l)_{l\in [L]},R^{(U)},(R^{(K)}_l)_{l\in [L]})$ private-sum computation protocol that is $(\delta= \alpha (L-1),T)$-private, where $\alpha \in [0,1]$, we~ have
\begin{align}
      R^{(X)}_l & \geq 1,  \forall l \in [L],  \label{eqth2a}\\
R^{(K)}_l & \geq 1-\alpha,   \forall l \in [L],\label{eqth2b}\\    R^{(U)} & \geq (1-\alpha) (L-1). \label{eqth2c}
\end{align}
\end{thm}
\begin{proof}
    See Section \ref{secthi}.
\end{proof}

\begin{rem}
   Suppose $\alpha=0$. Since \eqref{eqLsym} is always true when $\alpha = 0$, Theorem \ref{thmi} generalizes the converse in \cite{zhao2022secure}, as the local randomness rates and communication rates are not assumed to be necessarily  all equal in the derivation of our~converse.
\end{rem}

\subsection{Capacity results} \label{secca}

The following theorem shows that the converse bounds of Theorems \ref{thmcv} and \ref{thmi} can all be achieved simultaneously.
\begin{thm} \label{thmach}
For any $\alpha \in [0,1]\cap \mathbb{Q}$, we have
\begin{align*}
    R^{(K)}_{l,\star} & = 1- \alpha, \forall l\in [L],\text{ when \eqref{eqLsym} holds},\\
    R^{(K)}_{\Sigma,\star} & = (1- \alpha)L, \\
    R_{l,\star}^{(X)} & = 1, \forall l\in [L],\\
    R_{\star}^{(U)} & = (1-\alpha) (L-1).
\end{align*} 
Moreover, there exists a private-sum computation protocol that is $(\delta= \alpha (L-1),T)$-private and simultaneously achieves  $(R^{(X)}_{l,\star})_{l\in [L]}$, $R_{\star}^{(U)}$, $(R^{(K)}_{l,\star})_{l\in [L]}$, and $R^{(K)}_{\Sigma,\star}$.
\end{thm}
\begin{proof}
    See Section \ref{secach}.
\end{proof}

In Theorem \ref{thmach}, $\alpha$ is a rational number. However, note that by density of $\mathbb{Q}$ in $\mathbb{R}$, for any $\alpha' \in \mathbb{R}_+$, for any $\epsilon>0$, there exists $\alpha \in \mathbb{Q}_+$ such that $|\alpha - \alpha'| \leq \epsilon$.	

\begin{rem}
The  leakage symmetry assumption is only used to prove optimality of the individual rate of local randomness in Theorem \ref{thmach}.
Therefore,  in the absence of leakage symmetry, the following characterization obtained in Theorem~\ref{thmach} holds: 
    \begin{enumerate}[(i)]
        \item the optimal sum rate of local randomness,
        \item  the optimal individual rate of communication,
\item the optimal rate of global randomness.
    \end{enumerate} 
    These three results are valid for any $\alpha>0$, and thus subsumes previous results previously derived for $\alpha=0$, e.g.,~\cite{zhao2022secure}.
\end{rem}

\subsection{Connection between secret sharing and private summation} \label{secsc}

Before we establish the connection between the local randomness structure and secret sharing in Section \ref{secconn}, we review the notion of secret sharing in Section \ref{secprem}.

\subsubsection{Preliminaries}\label{secprem}
Uniform secret sharing with leakage, e.g.,~\cite{yoshida2018optimal,yoshida2012optimum,chou2020secure,chou2023secure}, has been defined as follows. 
\begin{defn} [Uniform secret sharing with leakage]
Let $\alpha \in [0,1]$, $t \in [L]$ and $z\in [t-1]$. An $(\alpha, t, z)$- secret sharing scheme consists of
\begin{itemize}
    \item  A secret $S$ uniformly distributed over $\{0, 1\}^{n_{s}};$
    \item  A stochastic encoder $ e : \{0, 1\}
^{n_{s}} \times \{0, 1\}^{n_{r}} \rightarrow \{0, 1\}^{n_{sh} },(S, R)\mapsto {( H_{l})_{l\in [L]}}$, which takes as input the secret $S$ and a randomization sequence $R$ uniformly distributed over $\{0, 1\}^{n_{r}}$ and independent of $S$, and outputs $L$ sequences ${( H_{l})_{l\in [L]}}$, with sum length  $n_{sh}$,  that are not necessarily of same length. ${( H_{l})_{l\in [L]}}$ are referred to as the shares of the secret sharing. For any $\mathcal{S} \subseteq [L]$, we define $H_{\mathcal{S}}={( H_{l})_{l\in \mathcal{S}}}; $ \end{itemize}
and satisfies the two conditions
\begin{align}
    \displaystyle \max_{\mathcal{T} \subseteq [L]:|\mathcal{T}|=t} H(S|H_{\mathcal{T}})&=0, \phantom{-}  \text{(Recoverability)} \label{marker12}\\
   \displaystyle \max_{\mathcal{U} \subseteq [L]:|\mathcal{U}|\leq z } I(S;H_{\mathcal{U}})&\leq \alpha H(S), \phantom{-}\text{(Privacy leakage)}\label{marker13}
\end{align} 
and the  leakage symmetry condition \begin{align} \label{eqlt}
\forall	l \in  [L],\exists C_{l} \in \mathbb{R}^+ , \forall \mathcal{T} \in [L], |\mathcal{T}|=l \! \implies \! \frac{I(S;H_{\mathcal{T}})}{H(S)} &= C_{l}.
\end{align}
\label{def}
\end{defn}

\begin{defn}
A $(t,\Delta,L)$  ramp secret sharing scheme, e.g., \cite{yamamoto1986secret,blakley1984security}, is an $(\alpha, t, z)$- secret sharing scheme with $\alpha =0$, $z = t - \Delta$, and optimal share sizes. 
 \end{defn}

Note that, by \cite{chou2023secure}, for a $(t,\Delta,L)$  ramp secret sharing scheme, we have
\begin{align}
C_l \triangleq \begin{cases}
0  & \text{if } l \in  \llbracket 0, t - \Delta \rrbracket   \\
 \frac{l - t + \Delta}{\Delta}    & \text{if } l \in  \llbracket t - \Delta +1, t  \rrbracket   \\
  1 & \text{if } l \in  \llbracket t +1 , L \rrbracket
 \end{cases},\label{eqrampsec}\end{align} meaning that any $t$ shares can reconstruct $S$, any set of shares less than or equal to $t -\Delta$ does not leak any information about $S$, and for sets of shares with cardinality in $ \llbracket t - \Delta +1, t  \rrbracket$, the leakage increases linearly with the set cardinality.
 \medskip
\subsubsection{Private summation and  secret sharing}\label{secconn}
While it was known that secret sharing can be employed to generate the local randomness at the users when $\alpha=0$ in our setting, e.g., \cite[Theorem 1]{zhao2022secure}, Theorem~\ref{th4} establishes a stronger connection  between private summation and secret sharing by showing that secret sharing is necessary to generate the local randomness needed for private summation. 

\begin{thm}[Necessary Condition] \label{th4}
Suppose that leakage symmetry, i.e.,~\eqref{eqLsym} is required,  and consider an optimal private-sum computation protocol, i.e., 
\begin{align*}
    R^{(K)}_l  &=R^{(K)}_{l,\star} , \forall l\in [L],\\
    \textstyle\sum_{l\in[L]}R^{(K)}_l  &=R^{(K)}_{\Sigma,\star} , \\
    R^{(X)}_l &=R_{l,\star}^{(X)} , \forall l\in [L],\\
    R^{(U)}  &= R_{\star}^{(U)} .
\end{align*} 
Then,  $(K_l)_{l\in [L]}$ must be the shares of a $(L-1, L-1 ,L)$ ramp secret sharing scheme where $U$ is the secret. 
\end{thm}
\begin{proof}
See Section \ref{appproofth4}.
\end{proof}

\begin{rem}
    When $\alpha =0$, \eqref{eqLsym} is always satisfied, hence Theorem \ref{th4} proves that secret sharing is always necessary to generate the local randomness needed for secure summation. The same result is proved for $\alpha >0$ but only under the leakage symmetry assumption \eqref{eqLsym}.
\end{rem}

\section{Proof of Theorem \ref{thmcv}} \label{seccv}

\subsection{Communication rate converse: Proof of \eqref{eqth21}}
For any $l\in[L]$, we have
\begin{align*}
    &nR_l^{(X)}\\
    & \stackrel{(a)}\geq H(X_l) \\
    & \stackrel{(b)}\geq H(X_l| S_{[L]\backslash \{l\}} ,K_{[L]\backslash \{l\}})\\
    & \geq  I(X_l;\Sigma_{[L]}| S_{[L]\backslash \{l\}} ,K_{[L]\backslash \{l\}}) \\
    & = H(\Sigma_{[L]}| S_{[L]\backslash \{l\}}, K_{[L]\backslash \{l\}})\\
& \phantom{--} - H(\Sigma_{[L]} |X_l ,S_{[L]\backslash \{l\}}, K_{[L]\backslash \{l\}})\\
        & \stackrel{(c)}= H(\Sigma_{[L]} ,S_{[L]\backslash \{l\}}| S_{[L]\backslash \{l\}} ,K_{[L]\backslash \{l\}}) \\
        & \phantom{--}- H(\Sigma_{[L]} |X_l ,S_{[L]\backslash \{l\}} ,K_{[L]\backslash \{l\}})\\
        & \stackrel{(d)}= H(\Sigma_{[L]} ,S_l ,S_{[L]\backslash \{l\}}| S_{[L]\backslash \{l\}} ,K_{[L]\backslash \{l\}}) \\
        & \phantom{--}- H(\Sigma_{[L]} |X_l ,S_{[L]\backslash \{l\}} ,K_{[L]\backslash \{l\}})\\
                & \stackrel{(e)}\geq H( S_{l}| S_{[L]\backslash \{l\}} ,K_{[L]\backslash \{l\}}) - H(\Sigma_{[L]} |X_l, S_{[L]\backslash \{l\}}, K_{[L]\backslash \{l\}})\\
    & \stackrel{(f)}= H(S_{l}| S_{[L]\backslash \{l\}} ,K_{[L]\backslash \{l\}}) \\
& \phantom{--}- H(\Sigma_{[L]} |X_{[L]} ,S_{[L]\backslash \{l\}} ,K_{[L]\backslash \{l\}})\\
        & \stackrel{(g)}= H(S_{l}| S_{[L]\backslash \{l\}} ,K_{[L]\backslash \{l\}}) \\
    & \stackrel{(h)}= H(S_l) \\
    & \stackrel{(i)}= n,
\end{align*}
where 
\begin{enumerate}[(a)]
    \item holds by Definition \ref{def1};
    \item  holds because conditioning reduces entropy;
    \item holds by the chain rule;
    \item holds because $S_l$ can be reconstructed from $(\Sigma_{[L]}, S_{[L]\backslash \{l\}})$;
        \item holds by the chain rule;
    \item holds because $X_{[L]\backslash \{ l\}}$ is a function of $(S_{[L]\backslash \{l\}}, K_{[L]\backslash \{l\}})$;
    \item holds because $$H(\Sigma_{[L]} |X_{[L]} ,S_{[L]\backslash \{l\}} ,K_{[L]\backslash \{l\}}) \leq H(\Sigma_{[L]} |X_{[L]} )=0,$$ where the equality holds by \eqref{eqrec};
        \item holds by  independence between $S_l$ and $(S_{[L]\backslash \{l\}} ,K_{[L]\backslash \{l\}})$;
    \item holds by uniformity of $S_l$.
\end{enumerate}

\subsection{Local randomness sum-rate converse: Proof of \eqref{eqth22}}
We start by proving the following lemma that upper bounds the sum of of the information leaked by the encoder output of each user about their individual sequence.

\begin{lem}
We have    \begin{align}
  \sum_{l\in[L]} I(S_l;X_l) \leq \alpha n L.  \label{eql3}
\end{align}
\end{lem}

\begin{proof}

For any $i\in[L]$, we have
\begin{align*}
&\alpha n(L-1)  \displaybreak[0]\\
&\stackrel{(a)}\geq
I(S_{[L]}; X_{[L]} |\Sigma_{[L]}  ) \displaybreak[0]\\
& \stackrel{(b)}=
\sum_{l\in[L]\backslash\{i\}} I(S_{l}; X_{[L]} |S_{[l-1]\backslash\{i\}},\Sigma_{[L]}  )  \\
&\phantom{--}+ I(S_{i}; X_{[L]} |S_{[L]\backslash\{i\}},\Sigma_{[L]}  )\\
& \stackrel{(c)} =
\sum_{l\in[L]\backslash\{i\}} I(S_{l}; X_{[L]} |S_{[l-1]\backslash\{i\}},\Sigma_{[L]}  ) \\
& \stackrel{(d)} \geq
\sum_{l\in[L]\backslash\{i\}} I(S_{l}; X_{l} |S_{[l-1]\backslash\{i\}},\Sigma_{[L]}  ) \\
& \stackrel{(e)}=
\sum_{l\in[L]\backslash\{i\}} I(S_{l}; X_{l} |S_{[l-1]\backslash\{i\}},\Sigma_{[l:L]\cup\{i\}}  ) \\
& \stackrel{(f)}= 
\sum_{l\in[L]\backslash\{i\}} [I(S_{l}; X_{l}, S_{[l-1]\backslash\{i\}} |\Sigma_{[l:L]\cup\{i\}}  ) \\
&\phantom{--}-I(S_{l}; S_{[l-1]\backslash\{i\}}| \Sigma_{[l:L]\cup\{i\}}  ) ]\\
& \stackrel{(g)} = 
\sum_{l\in[L]\backslash\{i\}} I(S_{l}; X_{l}, S_{[l-1]\backslash\{i\}} |\Sigma_{[l:L]\cup\{i\}}  ) \\
& \stackrel{(h)}=
\sum_{l\in[L]\backslash\{i\}} [I(S_{l}; X_{l}  |\Sigma_{[l:L]\cup\{i\}}  ) \\
&\phantom{--}+ I(S_{l}; S_{[l-1]\backslash\{i\}} |X_{l} ,\Sigma_{[l:L]\cup\{i\}}  )]\\
& \stackrel{(i)} =
\sum_{l\in[L]\backslash\{i\}} I(S_{l}; X_{l}  |\Sigma_{[l:L]\cup\{i\}}  )\\
& =
\sum_{l\in[L]\backslash\{i\}} [I(S_{l}; X_{l},  \Sigma_{[l:L]\cup\{i\}}  ) - I(S_{l}; \Sigma_{[l:L]\cup\{i\}}  )]\\
& \stackrel{(j)} =  
\sum_{l\in[L]\backslash\{i\}} I(S_{l}; X_{l},  \Sigma_{[l:L]\cup\{i\}}  ) \\
& \geq  
\sum_{l\in[L]\backslash\{i\}} I(S_{l}; X_{l}  ), \numberthis \label{eql2}
\end{align*}
where 
\begin{enumerate}[(a)]
    \item holds by \eqref{eqprivact2} with $\mathcal{T} = 0$;
    \item holds by the chain rule;
    \item holds because $S_i$ can be recovered from $(S_{[L]\backslash\{i\}},\Sigma_{[L]})$;
        \item holds by the chain rule;
    \item holds with the notation \begin{align*}\Sigma_{[l:L]\cup \{i\}} \triangleq \sum_{j \in \{l, l+1, \ldots,L\}\cup \{i\}}S_j;\end{align*}
        \item holds by the chain rule;
    \item holds because \begin{align*}
        &I(S_{l}; S_{[l-1]\backslash\{i\}}| \Sigma_{[l:L]\cup\{i\}}  )  \\
        &\leq I(S_{l},\Sigma_{[l:L]\cup\{i\}}; S_{[l-1]\backslash\{i\}} ) \\&=0\end{align*}where the inequality holds by the chain rule and non-negativity of mutual information, and the equality holds because $(S_{l},\Sigma_{[l:L]\cup\{i\}})$ is a deterministic function of  $S_{\llbracket l,L\rrbracket \cup\{i\}}$, which is independent of $S_{[l-1]\backslash\{i\}} $ by independence of the users’ sequences, since $(\llbracket l,L\rrbracket\cup\{i\} )\cap ([l-1]\backslash\{i\})=  \emptyset$; 
                \item holds by the chain rule;
    \item holds because
    \begin{align*}
        &I(S_{l}; S_{[l-1]\backslash\{i\}} |X_{l} ,\Sigma_{[l:L]\cup\{i\}})  \\&\leq I(S_{l}, \Sigma_{[l:L]\cup\{i\}},X_{l}; S_{[l-1]\backslash\{i\}}  )\\&\leq I(S_{l}, \Sigma_{[l:L]\cup\{i\}},K_{l}; S_{[l-1]\backslash\{i\}}  )\\&\leq I(S_{l} ,\Sigma_{[l:L]\cup\{i\}},U; S_{[l-1]\backslash\{i\}}  )\\ &= I(S_{l}, \Sigma_{[l:L]\cup\{i\}}; S_{[l-1]\backslash\{i\}}  )\\&=0,\end{align*} where the first inequality holds by the chain rule and non-negativity of the mutual information, the second  inequality holds because $X_l$ is a function of $(S_l,K_l)$, the third  inequality holds because $K_l$ is a function of $U$,  the first equality holds by independence of the users' sequences and the global randomness, and the last equality holds as in (g);
    \item holds by the one-time pad lemma, e.g., \cite[Th. 2.4]{stinson2005cryptography}, using the uniformity of the users' sequences and that $l\neq i$. 
\end{enumerate}

Hence, by remarking that
\begin{align}
    \sum_{i\in [L]}\sum_{l\in[L]\backslash \{ i\}} I(S_l;X_l) = (L-1) \sum_{l\in[L]} I(S_l;X_l), \label{eql1}
\end{align}
from   \eqref{eql2} and \eqref{eql1}, we obtain \eqref{eql3}.

\end{proof}

Then, we have
\begin{align*}
  \sum_{l\in[L]}  nR^{(K)}_{l} 
    & \stackrel{(a)}\geq \sum_{l\in [L]} H(K_l) \\
    & \stackrel{(b)} = \sum_{l\in [L]}  H(K_l|S_l)    \displaybreak[0]\\
    & \geq \sum_{l\in [L]}  I(K_l;X_l| S_{l})   \\
    & = \sum_{l\in [L]} [H(X_l|S_l)  - H(X_l|K_l ,S_l) ]\\
    & \stackrel{(c)}= \sum_{l\in [L]} H(X_l|S_l)  \\
    & = \sum_{l\in [L]}[H(X_l) - I(X_l;S_l)] \\
    & \stackrel{(d)}\geq nL - \sum_{l\in [L]}I(X_l;S_l) \\
    & \stackrel{(e)}\geq nL - \alpha nL\\
    & = nL(1 - \alpha),
\end{align*}
where 
\begin{enumerate}[(a)]
    \item holds by Definition \ref{def1};
    \item holds by independence between $U$ and $S_{[L]}$;
    \item holds by Definition \ref{def1}; 
    \item  holds by~\eqref{eqth21};
    \item holds by~\eqref{eql3}.
\end{enumerate}

\subsection{Global randomness rate converse: Proof of \eqref{eqth23}}

We have
\begin{align*}
    &nR^{(U)} \\
    &\stackrel{(a)}\geq H(U) \\
    & \stackrel{(b)}\geq H(K_{[L]}) \\
    & \stackrel{(c)}\geq H(K_{[L]}| S_{[L]}) \\
    & \geq I(K_{[L]}; X_{[L]} | S_{[L]}) \\
    & \stackrel{(d)}= H(X_{[L]} | S_{[L]}) \\
    & = H(X_{[L]} ) - I( X_{[L]}; S_{[L]}) \\
    & \stackrel{(e)}= H(X_{[L]} ) - I( X_{[L]}; \Sigma_{[L]}, S_{[L]}) \\
    & = H(X_{[L]} ) - I( X_{[L]}; \Sigma_{[L]} ) -I( X_{[L]};  S_{[L]}|\Sigma_{[L]})\\
    & \stackrel{(f)}\geq H(X_{[L]} ) - I( X_{[L]}; \Sigma_{[L]} ) - \alpha n (L-1)\\
    & = H(X_{[L]} ) - H(\Sigma_{[L]}) + H( \Sigma_{[L]}| X_{[L]} ) - \alpha n (L-1)\\
    & \stackrel{(g)}= H(X_{[L]} ) - H(\Sigma_{[L]})  - \alpha n (L-1)\\
    & \stackrel{(h)}= H(X_{[L]} ) - n - \alpha n (L-1)\\
    & = \sum_{l \in [L]} H(X_{l} |X_{[l-1]} ) - n - \alpha n (L-1)\\
    & \stackrel{(i)}\geq \sum_{l \in [L]} H(X_{l} |S_{[L]\backslash \{ l\}} ,K_{[L]\backslash \{ l\}} ,X_{[l-1]}) - n - \alpha n (L-1)\\
    & \stackrel{(j)}= \sum_{l \in [L]} H(X_{l} |S_{[L]\backslash \{ l\}}, K_{[L]\backslash \{ l\}} ) - n - \alpha n (L-1)\\
    & \geq \sum_{l \in [L]} I(X_{l} ; \Sigma_{[L]} |S_{[L]\backslash \{ l\}} ,K_{[L]\backslash \{ l\}} ) - n - \alpha n (L-1)\\
    & = \sum_{l \in [L]} [H(\Sigma_{[L]}  |S_{[L]\backslash \{ l\}} ,K_{[L]\backslash \{ l\}} ) \\
&\phantom{--}- H(\Sigma_{[L]}|X_{l} , S_{[L]\backslash \{ l\}},K_{[L]\backslash \{ l\}} )] - n - \alpha n (L-1)\\
    & \stackrel{(k)}= \sum_{l \in [L]} H(\Sigma_{[L]}  |S_{[L]\backslash \{ l\}} ,K_{[L]\backslash \{ l\}} )  - n - \alpha n (L-1)\\
        & \stackrel{(l)}= \sum_{l \in [L]} H(\Sigma_{[L]} ,S_l |S_{[L]\backslash \{ l\}} ,K_{[L]\backslash \{ l\}} )  - n - \alpha n (L-1)\\
    & \geq \sum_{l \in [L]} H(S_l |S_{[L]\backslash \{ l\}}, K_{[L]\backslash \{ l\}} )  - n - \alpha n (L-1)\\
    & \stackrel{(m)}= \sum_{l \in [L]} H(S_l )  - n - \alpha n (L-1)\\
    & \stackrel{(n)}=  n(L-1)(1- \alpha),
\end{align*}
where \begin{enumerate}[(a)]
    \item holds by Definition \ref{def1};
    \item holds because $K_{[L]}$ is a function of $U$;
    \item holds because conditioning reduces entropy; 
    \item holds because $X_{[L]}$ is a function of $(K_{[L]}, S_{[L]})$;
    \item holds because $\Sigma_{[L]}$ is a function of $S_{[L]}$;
    \item holds by \eqref{eqprivact2} with $\mathcal{T} = 0$;
    \item holds by \eqref{eqrec};
    \item holds by uniformity of the users' sequences;
    \item holds because conditioning reduces entropy; 
    \item holds because $X_{[l-1]}$ is a function of $(S_{[L]\backslash \{ l\}}, K_{[L]\backslash \{ l\}})$; \item holds by $\eqref{eqrec}$ since $X_{[L]}$ can be reconstructed from $(X_{l},  S_{[L]\backslash \{ l\}} ,K_{[L]\backslash \{ l\}} )$; 
    \item holds by the chain rule and because $S_l$ can be reconstructed from $(\Sigma_{[L]},S_{[L]\backslash \{ l\}})$; 
    \item holds by independence between the users' sequences and the local randomness;
    \item  holds by uniformity of the users' sequences. 
\end{enumerate}

\section{Proof of Theorem \ref{thmi}} \label{secthi}
Note that the proof of \eqref{eqth2a} and \eqref{eqth2c} follows from Theorem~\ref{thmcv}, for which the leakage symmetry \eqref{eqLsym} is not needed. Hence, it is sufficient to prove \eqref{eqth2b} under the  leakage symmetry \eqref{eqLsym}.

In the following lemma, we first give a general characterization of the leakage associated with an arbitrary set of colluding~users.

\begin{lem} \label{lem1}
For any $\mathcal{T} \subseteq [L]$, we have
\begin{align*}
    C_{|\mathcal{T}|} 
   = n(L- |\mathcal{T}| -\mathds{1}\{ \mathcal{T} \neq [L]\}) - H(S_{\mathcal{T}^c}| X_{[L]} K_{\mathcal{T}}   S_{\mathcal{T}} ).
\end{align*}
\end{lem}
\begin{proof}
See Appendix \ref{applem1}.
\end{proof}

In the following lemma, we provide a characterization of the difference between the leakages associated with any two sets of colluding users that differ in size by only one member.

\begin{lem} \label{lem2}
Fix $l\in[L]$. For $t\in \llbracket 0,L-1 \rrbracket$, define 
$$ 
\mathcal{T}_t \triangleq \begin{cases} \llbracket 1,t \rrbracket & \text{ if } t<l \\ \llbracket 1,t +1\rrbracket \backslash \{l\} & \text{ if } t \geq l \end{cases},
$$
and $\mathcal{T}_{L} \triangleq [L]$. 

Then, for $t \in \llbracket 0, L-1 \rrbracket$, we have
\begin{align*}
    &C_{t+1} - C_t\\
      &=H(K_l ,S_l | K_{\mathcal{T}_t},S_{\mathcal{T}_t},X_{[L]})- H(K_l | K_{\mathcal{T}_t},S_{[L]},X_{[L]}) \\
      &\phantom{-}- n \mathds{1}\{t\neq L\!-\!1\}. \numberthis\label{eqlem3}
\end{align*}
\end{lem}
\begin{proof}
For $t \in \llbracket 0, L-2 \rrbracket$, we have
\begin{align*}
    &  H(K_l ,S_l | K_{\mathcal{T}_t},S_{\mathcal{T}_t},X_{[L]}) \\
    & \stackrel{(a)}=H(K_l ,S_l,S_{\mathcal{T}_t^c}| K_{\mathcal{T}_t},S_{\mathcal{T}_t},X_{[L]}) \\
&\phantom{--}- H( S_{\mathcal{T}_t^c}|  K_{\mathcal{T}_t\cup\{l\}},S_{\mathcal{T}_t\cup\{l\}},X_{[L]}) \\
    & \stackrel{(b)}=H(S_{\mathcal{T}_t^c}| K_{\mathcal{T}_t},S_{\mathcal{T}_t},X_{[L]})+H(K_l,S_l |K_{\mathcal{T}_t},S_{[L]},X_{[L]}) \\
&\phantom{--}- H( S_{\mathcal{T}_t^c}|  K_{\mathcal{T}_t\cup\{l\}},S_{\mathcal{T}_t\cup\{l\}},X_{[L]}) \\
    & \stackrel{(c)}  = n(L-t-1)- C_t+H(K_l,S_l |  K_{\mathcal{T}_t},S_{[L]},X_{[L]}) \\
&\phantom{--}- H( S_{\mathcal{T}_t^c}|  K_{\mathcal{T}_t\cup\{l\}},S_{\mathcal{T}_t\cup\{l\}},X_{[L]}) \\
    &   \stackrel{(d)}= n(L-t-1)- C_t+H(K_l | K_{\mathcal{T}_t},S_{[L]},X_{[L]}) \\
&\phantom{--}- H( S_{\mathcal{T}_t^c \backslash \{l\}}|K_{\mathcal{T}_t\cup\{l\}},S_{\mathcal{T}_t\cup\{l\}},X_{[L]}) \\
    & \stackrel{(e)} = n(L-t-1)- C_t+H(K_l | K_{\mathcal{T}_t},S_{[L]},X_{[L]}) + C_{t+1} \\
&\phantom{--}- n(L-(t+1)-1) \\
    & = n + C_{t+1} - C_t+H(K_l | K_{\mathcal{T}_t},S_{[L]},X_{[L]})   , \numberthis\label{eqlem3a}
\end{align*}
where 
\begin{enumerate}[(a)]
    \item\!\!\!, (b), and (d) hold by the chain rule;
    \addtocounter{enumi}{1}
    \item and (e) hold by Lemma \ref{lem1}. 
\end{enumerate}

Moreover, we have
\begin{align*}
  &H(K_l ,S_l | K_{\mathcal{T}_{L-1}},S_{\mathcal{T}_{L-1}},X_{[L]}) \\
    & \stackrel{(a)}=H(K_l , S_l| K_{\mathcal{T}_{L-1}},S_{[L]},X_{[L]})\\
        & =H(K_l  | K_{\mathcal{T}_{L-1}},S_{[L]},X_{[L]})\\
    & \stackrel{(b)}=  C_{L} - C_{L-1}+ H(K_l  | K_{\mathcal{T}_{L-1}},S_{[L]},X_{[L]})   , \numberthis\label{eqlem3b}
\end{align*}
where \begin{enumerate}[(a)]
    \item holds because $S_{[L]}$ can be reconstructed from $X_{[L]}$ by \eqref{eqrec};
    \item holds because $C_{L-1}=0=C_{L}$.
\end{enumerate} 
Hence, \eqref{eqlem3} holds by \eqref{eqlem3a} and \eqref{eqlem3b}.
\end{proof}
Then, we lower bound the individual local randomness of User $l\in[L]$ as follows:
\begin{align*}
&H(K_l) \\
& \stackrel{(a)}\geq H(K_l)+ n -H(X_l) \\
& \stackrel{(b)}=H(K_l)+ H( S_l ) -H(X_l) \\
& \stackrel{(c)}=H(K_l ,S_l ) -H(X_l) \\
& \stackrel{(d)}=H(K_l ,S_l ,X_l) -H(X_l) \\
& = H(K_l, S_l|X_l) \\
& \stackrel{(e)}\geq   H(K_l ,S_l|K_{\mathcal{T}_0},S_{\mathcal{T}_0},X_{[L]}) \\
& \geq   H(K_l ,S_l|K_{\mathcal{T}_0},S_{\mathcal{T}_0},X_{[L]})   \\
& \phantom{--}- H(K_l ,S_l|K_{\mathcal{T}_{L-1}},S_{\mathcal{T}_{L-1}},X_{[L]})\\
    & = \sum_{i=0}^{L-2}   [H(K_l,S_l|K_{\mathcal{T}_i},S_{\mathcal{T}_i},X_{[L]}) \\
& \phantom{----}- H(K_l,S_l|K_{\mathcal{T}_{i+1}},S_{\mathcal{T}_{i+1}},X_{[L]})]\\
   &\stackrel{(f)}= \sum_{i=0}^{L-3} [C_{i+1} - C_i -C_{i+2} + C_{i+1} +H(K_l | K_{\mathcal{T}_i},S_{[L]},X_{[L]})\\
&\phantom{--} - H(K_l | K_{\mathcal{T}_{i+1}},S_{[L]},X_{[L]})]\\
   & \phantom{--}+ [n+C_{L-1} - C_{L-2}  +H(K_l | K_{\mathcal{T}_{L-2}},S_{[L]},X_{[L]})  \\
&\phantom{---}- C_{L}+ C_{L-1}- H(K_l  | K_{\mathcal{T}_{L-1}},S_{[L]},X_{[L]}) ]\\
      & \stackrel{(g)}\geq n+ \sum_{i=0}^{L-2} [C_{i+1} - C_i -C_{i+2} + C_{i+1} ] \\
  & = n+  C_1- C_0 - C_{L} + C_{L-1}\\
   & \stackrel{(h)}= C_1 - C_0 +n\\
   & \stackrel{(i)}\geq n(1-\alpha),
\end{align*}
where \begin{enumerate}[(a)]
    \item holds by \eqref{eqth21};
    \item holds by uniformity of the users' sequences;
    \item holds by independence between the users' sequences and the local randomness;
    \item holds by Definition~\ref{def1};
    \item holds because conditioning reduces entropy; 
    \item holds by Lemma \ref{lem2}; 
    \item holds because, for any $i \in \llbracket 0, L-2 \rrbracket$, \begin{align*}    
    H(K_l | K_{\mathcal{T}_i},S_{[L]},X_{[L]}) \geq H(K_l | K_{\mathcal{T}_{i+1}},S_{[L]},X_{[L]})\end{align*}  since conditioning reduces entropy; 
    \item follows from $C_{L-1}=0=C_{L}$;
    \item holds by~Property \ref{prop3}.
    \end{enumerate}

\section{Proof of Theorem \ref{thmach}} \label{secach}
To prove Theorem \ref{thmach}, it is sufficient to prove the following achievability result, which proves that the converse bounds of Section \ref{secconverse} are tight.
\begin{prop}[Achievability]
For any $\alpha \in [0,1]\cap \mathbb{Q}$, there exists an $(L,n,(R^{(X)}_l)_{l\in [L]},R^{(U)},(R^{(K)}_l)_{l\in [L]})$ private-sum computation protocols that is $(\delta= \alpha (L-1),T)$-private and such that
\begin{align*}
    R^{(K)}_l & = (1- \alpha), \forall l\in [L],\\
    R^{(X)}_l & = 1, \forall l\in [L],\\
    R^{(U)} & = (1-\alpha) (L-1).
\end{align*} 
\end{prop}
\begin{proof}
The achievability scheme can be seen as a time-sharing version of the coding scheme in~\cite{zhao2022secure}.  Let $n_1, n_2 \in \mathbb{N}^*$ and define $n\triangleq n_1 + n_2$, $\alpha \triangleq n_1 /n$. For $l\in[L]$, the sequence $S_l$ is written as $S_l = (S_l^{(1)},S_l^{(2)})$ where $|S_l^{(1)}| = n_1$ and $|S_l^{(2)}| = n_2$. User $l\in [L]$ sends $X_l \triangleq (S^{(1)}_l,K_l + S^{(2)}_l)$ where $(K_l)_{l\in[L-1]}$ are jointly uniform sequences of $n_2$ bits, independent of $S_{[L]}$, and $K_{L} \triangleq \sum_{l\in[L-1]}K_l$. Hence, we have for any $l\in[L]$, $R^{(K)}_l  = \frac{n_2}{n} = 1- \alpha$, $R^{(X)}_l =\frac{n}{n}=1$, and $R^{(U)} = \frac{(L-1)n_2}{n}=(L-1)(1-\alpha)$.
Then, for any $\mathcal{T} \subset [L]$ such that $|\mathcal{T}|=T$, we~have
\begin{align*}
&I(S_{[L]}; X_{[L]}  |\Sigma_{[L]} ,S_{\mathcal{T}} ,K_{\mathcal{T}})\\
& \stackrel{(a)}= I(S_{\mathcal{T}^c}; X_{\mathcal{T}^c}  |\Sigma_{\mathcal{T}^c} ,S_{\mathcal{T}} ,K_{\mathcal{T}}) \\
& = H(X_{\mathcal{T}^c}|\Sigma_{\mathcal{T}^c} ,S_{\mathcal{T}}, K_{\mathcal{T}}) - H(X_{\mathcal{T}^c} | S_{ [L]},   \Sigma_{\mathcal{T}^c} ,  K_{\mathcal{T}})  \\
& \stackrel{(b)} = H(X_{\mathcal{T}^c}|\textstyle\sum_{l\in \mathcal{T}^c} X_l, \Sigma_{\mathcal{T}^c} ,S_{\mathcal{T}} ,K_{\mathcal{T}})\\
& \phantom{--} - H(X_{\mathcal{T}^c} | S_{ [L]},   \Sigma_{\mathcal{T}^c} ,  K_{\mathcal{T}})  \\
& \stackrel{(c)}\leq H(X_{\mathcal{T}^c}| \textstyle\sum_{l\in \mathcal{T}^c} X_l ) - H(X_{\mathcal{T}^c} | S_{ [L]} ,  \Sigma_{\mathcal{T}^c} ,  K_{\mathcal{T}}) \\
& \stackrel{(d)}=H(X_{\mathcal{T}^c}| \textstyle\sum_{l\in \mathcal{T}^c} X_l  ) - H(K_{\mathcal{T}^c} |   S_{ [L]} , \Sigma_{\mathcal{T}^c},   K_{\mathcal{T}})  \\
& \stackrel{(e)}=H(X_{\mathcal{T}^c}| \textstyle\sum_{l\in \mathcal{T}^c} X_l ) - H(K_{\mathcal{T}^c} | K_{\mathcal{T}})  \\
& =  (L - |\mathcal{T}| -1) n - (L - |\mathcal{T}| -1) n (1-\alpha) \\
& =  (L - |\mathcal{T}| -1) n \alpha \\
& \leq (L -1) n \alpha,
\end{align*}
where \begin{enumerate}[(a)]
\item holds by the chain rule because $X_{\mathcal{T}}$ can be recovered from $(S_{\mathcal{T}},K_{\mathcal{T}})$, and there is a one-to-one mapping between $(\Sigma_{\mathcal{T}^c} ,S_{\mathcal{T}})$ and $(\Sigma_{[L]}, S_{\mathcal{T}})$; 
    \item holds because $\sum_{ \in \mathcal{T}^c}K_l $ can be computed from $K_{\mathcal{T}}$ as $\sum_{ \in \mathcal{T}^c}K_l  = \sum_{ \in \mathcal{T}}K_l$ by definition of $K_L$; \item  holds because conditioning reduces entropy; \item holds because \begin{align*}
        &H(X_{\mathcal{T}^c} | S_{ [L]}   ,\Sigma_{\mathcal{T}^c}   ,K_{\mathcal{T}})\\&= H(X_{\mathcal{T}^c} ,S_{ [L]}| S_{ [L]}   ,\Sigma_{\mathcal{T}^c}   ,K_{\mathcal{T}})\\&= H(K_{\mathcal{T}^c} ,S_{ [L]}| S_{ [L]}   ,\Sigma_{\mathcal{T}^c}   ,K_{\mathcal{T}})\\&= H(K_{\mathcal{T}^c} | S_{ [L]},   \Sigma_{\mathcal{T}^c}   ,K_{\mathcal{T}})\end{align*} by the definition of $X_l$, $l\in[L]$; \item  holds by independence between $K_{[L]}$ and $S_{[L]}$.
    \end{enumerate}
\end{proof}

\section{Proof of Theorem \ref{th4}} \label{appproofth4}
We start by deriving a lower bound on the minimum amount of joint communication from any set of users in the following~lemma.
\begin{lem} \label{lemcc}
For $\mathcal{T} \subseteq [L]$, we have
\begin{align*}
    H(X_{\mathcal{T}})
    \geq n|\mathcal{T}|.
\end{align*}
\end{lem}
\begin{proof}
    For $\mathcal{T} \subseteq [L]$ such that $|\mathcal{T}|\geq 1$ and $t \in \mathcal{T}$, we have
\begin{align*}
    &H(X_{\mathcal{T}})\\
    &\stackrel{(a)}\geq H(X_{\mathcal{T}}| S_{[L]\backslash \mathcal{T}},K_{[L]\backslash \mathcal{T}}) \numberthis \label{eqi1} \\
    & \geq I(X_{\mathcal{T}}; \Sigma_{[L]},X_{\mathcal{T}\backslash\{t\}}| S_{[L]\backslash \mathcal{T}},K_{[L]\backslash \mathcal{T}}) \\
    & =  H( \Sigma_{[L]},X_{\mathcal{T}\backslash\{t\}}| S_{[L]\backslash \mathcal{T}},K_{[L]\backslash \mathcal{T}}) \\
&\phantom{--}- H( \Sigma_{[L]},X_{\mathcal{T}\backslash\{t\}}| X_{\mathcal{T}}, S_{[L]\backslash \mathcal{T}},K_{[L]\backslash \mathcal{T}})\\
    & \stackrel{(b)} =  H( \Sigma_{[L]},X_{\mathcal{T}\backslash\{t\}}| S_{[L]\backslash \mathcal{T}},K_{[L]\backslash \mathcal{T}})\\
    &  =  H( X_{\mathcal{T}\backslash\{t\}}| S_{[L]\backslash \mathcal{T}},K_{[L]\backslash \mathcal{T}}) \\
&\phantom{--}+ H( \Sigma_{[L]}|X_{\mathcal{T}\backslash\{t\}} ,S_{[L]\backslash \mathcal{T}},K_{[L]\backslash \mathcal{T}}) \\
    & \stackrel{(c)} \geq   H( X_{\mathcal{T}\backslash\{t\}}| S_{[L]\backslash (\mathcal{T}\backslash\{t\})},K_{[L]\backslash (\mathcal{T}\backslash\{t\})}) \\
&\phantom{--}+ H( \Sigma_{[L]}|X_{\mathcal{T}\backslash\{t\}},S_{[L]\backslash\{t\}} ,K_{[L]\backslash\{t\}} ) \\
    & \stackrel{(d)} =   H( X_{\mathcal{T}\backslash\{t\}}| S_{[L]\backslash (\mathcal{T}\backslash\{t\})},K_{[L]\backslash (\mathcal{T}\backslash\{t\})}) \\
&\phantom{--}+ H( \Sigma_{[L]}|S_{[L]\backslash\{t\}} ,K_{[L]\backslash\{t\}} ) \\
    & \stackrel{(e)} \geq    H( X_{\mathcal{T}\backslash\{t\}}| S_{[L]\backslash (\mathcal{T}\backslash\{t\})},K_{[L]\backslash (\mathcal{T}\backslash\{t\})}) + n \numberthis \label{eqi2}\\
    & \stackrel{(f)} \geq n|\mathcal{T}|, 
\end{align*}
where \begin{enumerate}[(a)]
    \item holds because conditioning reduces entropy; \item holds because \begin{align*}
        &H( \Sigma_{[L]},X_{\mathcal{T}\backslash\{t\}}| X_{\mathcal{T}} ,S_{[L]\backslash \mathcal{T}},K_{[L]\backslash \mathcal{T}})\\&= H( \Sigma_{[L]} | X_{\mathcal{T}} ,S_{[L]\backslash \mathcal{T}},K_{[L]\backslash \mathcal{T}})\\& \leq H( \Sigma_{[L]} | X_{[L]}) \\&= 0,\end{align*} where the inequality holds because $X_{[L]\backslash \mathcal{T}}$ is a function of $(S_{[L]\backslash \mathcal{T}},K_{[L]\backslash \mathcal{T}})$ and the equality holds by \eqref{eqrec}; \item holds because conditioning reduces entropy; \item holds because $X_{\mathcal{T}\backslash\{t\}}$ can be recovered from $(S_{\mathcal{T}\backslash\{t\}},K_{\mathcal{T}\backslash\{t\}})$; \item holds because \begin{align*}
        &H( \Sigma_{[L]}|S_{[L]\backslash\{t\}} ,K_{[L]\backslash\{t\}} )\\&=H( \Sigma_{[L]},S_t|S_{[L]\backslash\{t\}} ,K_{[L]\backslash\{t\}} )\\
        &\phantom{--}-H( S_t|\Sigma_{[L]},S_{[L]\backslash\{t\}} ,K_{[L]\backslash\{t\}} )\\&=H( \Sigma_{[L]},S_t|S_{[L]\backslash\{t\}} ,K_{[L]\backslash\{t\}} )\\& = H( S_t|S_{[L]\backslash\{t\}} ,K_{[L]\backslash\{t\}} )\\&= H( S_t)\\&=n,\end{align*} where the second equality holds because $S_t$ can be recovered from $(S_{[L]\backslash\{t\}},\Sigma_{[L]})$, the third equality holds by the chain rule and because $\Sigma_{[L]}$ can be recovered from $(S_{[L]\backslash\{t\}},S_t)$, the fourth  equality holds by independence between the users' sequences and the local randomness, and the last equality holds by uniformity of the users' sequences; \item holds by  repeating $|\mathcal{T}|-1$ times the steps between \eqref{eqi1} and \eqref{eqi2}. 
    \end{enumerate}
\end{proof}
We then use Lemma \ref{lemcc} to prove a lower bound on the joint amount of required local randomness at any set of users in the following~lemma.
\begin{lem} \label{lem4}
For any $\mathcal{T} \subseteq [L]$, we have
    \begin{align*}
    H(K_{\mathcal{T}})
   \geq  n (|\mathcal{T}|-\mathds{1}\{\mathcal{T} = [L]\}) (1 - \alpha). 
    \end{align*}
\end{lem}
\begin{proof}
    Fix $\mathcal{T} \subseteq [L]$ such that $|\mathcal{T}|\leq L - 2$, write its complement as $\mathcal{T}^c = (a_t)_{t \in [L-|\mathcal{T}|]}$, and define for $t \in [L-|\mathcal{T}|]$ $\mathcal{A}_t \triangleq (a_i)_{i\in[t]}$. Then, for $t=1$, we have
    \begin{align*}
    &H(K_{\mathcal{T}})\\
    & \stackrel{(a)} \geq H(K_{\mathcal{T}}) + |\mathcal{T}| n - H(X_{\mathcal{T}}) \\
    & \stackrel{(b)} = H(K_{\mathcal{T}}) + H(S_{\mathcal{T}}) - H(X_{\mathcal{T}})\\
    & \stackrel{(c)} = H(K_{\mathcal{T}} ,S_{\mathcal{T}}) - H(X_{\mathcal{T}})\\
    & \stackrel{(d)}= H(K_{\mathcal{T}} ,S_{\mathcal{T}},X_{\mathcal{T}}) - H(X_{\mathcal{T}})\\
    & = H(K_{\mathcal{T}} ,S_{\mathcal{T}}|X_{\mathcal{T}}) \numberthis \label{eqlem4}\\
    & \stackrel{(e)}\geq H(K_{\mathcal{T}} ,S_{\mathcal{T}}|K_{\mathcal{A}_t},S_{\mathcal{A}_t},X_{[L]}) \\
    & = H(K_{\mathcal{T}} ,S_{\mathcal{T}},S_{\mathcal{A}_t^c}|K_{\mathcal{A}_t},S_{\mathcal{A}_t},X_{[L]}) \\
&\phantom{--}- H(S_{\mathcal{A}_t^c}|K_{\mathcal{A}_t\cup \mathcal{T}},S_{\mathcal{A}_t \cup \mathcal{T}},X_{[L]}) \\
    & \geq H(S_{\mathcal{A}_t^c}|K_{\mathcal{A}_t},S_{\mathcal{A}_t},X_{[L]})  - H(S_{\mathcal{A}_t^c}|K_{\mathcal{A}_t\cup \mathcal{T}},S_{\mathcal{A}_t \cup \mathcal{T}},X_{[L]}) \\
    & \stackrel{(f)}= H(S_{\mathcal{A}_t^c}|K_{\mathcal{A}_t},S_{\mathcal{A}_t},X_{[L]})  \\
&\phantom{--}- H(S_{(\mathcal{A}_t \cup \mathcal{T})^c}|K_{\mathcal{A}_t\cup \mathcal{T}},S_{\mathcal{A}_t \cup \mathcal{T}},X_{[L]})\\
    &\stackrel{(g)}= n(L-t-1) -C_t -n(L-(t+|\mathcal{T}|)-1) + C_{t+|\mathcal{T}|} \\
    &\stackrel{(h)}\geq n |\mathcal{T}| (1 - \alpha), 
    \end{align*}
    where \begin{enumerate}[(a)]
    \item holds by Lemma \ref{lemcc}; \item holds by uniformity of the users' sequences; \item holds by independence between the users' sequences and the local randomness; \item holds because $X_{\mathcal{T}}$ can be recovered from $(K_{\mathcal{T}}, S_{\mathcal{T}})$;\item holds because conditioning reduces entropy; \item holds because $\mathcal{A}_t^c= (\mathcal{A}_t^c \cap \mathcal{T}^c)\cup(\mathcal{A}_t^c \cap \mathcal{T})$ and by the chain rule;
    \item holds by Lemma \ref{lem1}; \item holds because $C_{t+|\mathcal{T}|} - C_t = -\Delta|\mathcal{T}|$ by Property \ref{prop1}, and by Property \ref{prop3}.\end{enumerate} Then, for $\mathcal{T} \subseteq [L]$ such that $|\mathcal{T}|= L-1$, we have
    \begin{align*}
 H(K_{\mathcal{T}})
    & \stackrel{(a)} \geq H(K_{\mathcal{T}} ,S_{\mathcal{T}}|X_{\mathcal{T}}) \\
    & \stackrel{(b)}\geq H(S_{\mathcal{T}}|X_{[L]})\\
    & \stackrel{(c)} =  H(S_{\mathcal{T}} ,\Sigma_{[L]}|X_{[L]})\\
  & \stackrel{(d)}=  H(S_{[L]} |X_{[L]})\\
  & \stackrel{(e)} = n(L-1) - C_0 \\
  & \stackrel{(f)} \geq n (L-1) (1-\alpha ), \numberthis \label{lem42}
    \end{align*}
where \begin{enumerate}[(a)]
    \item holds similar to \eqref{eqlem4}; \item holds because conditioning reduces entropy; \item holds by \eqref{eqrec}; \item holds because $|\mathcal{T}|= L-1$; \item holds by Lemma \ref{lem1}; \item holds by Property \ref{prop3}. \end{enumerate}

    Finally, we have \begin{align*}    
    H(K_{[L]}) &\geq H(K_{[L-1]}) \\&\geq n (L-1) (1-\alpha ),\end{align*} where the last inequality holds by \eqref{lem42}.
\end{proof}
Using Lemma \ref{lem4}, we now characterize the   entropy of the global randomness given the local randomness of any strict subset of users.

\begin{lem} \label{lem6}
    For any $\mathcal{T}\subsetneq [L]$, we have 
\begin{align*}
    H(U|K_{\mathcal{T}}) 
    & =  n (L-|\mathcal{T}|-1) (1-\alpha ), \numberthis \label{secrets1}
\end{align*}
\end{lem}

\begin{proof}
For any $\mathcal{T}\subsetneq [L]$, we have 
\begin{align*}
    H(U|K_{\mathcal{T}}) 
    & \stackrel{(a)}= H(U) - H(K_{\mathcal{T}})\\
    & \stackrel{(b)}= n (L-1) (1-\alpha ) - H(K_{\mathcal{T}})\\
    & \stackrel{(c)}= n (L-1) (1-\alpha ) - n|\mathcal{T}| (1-\alpha)\\
    & =  n (L-|\mathcal{T}|-1) (1-\alpha ), 
\end{align*}
where \begin{enumerate}[(a)]
\item holds because $K_{\mathcal{T}}$ is a function of $U$;
    \item holds by optimality of the private-sum computation protocol; \item holds because $H(K_{\mathcal{T}}) \leq \sum_{t\in \mathcal{T}}H(K_{t}) = n|\mathcal{T}| (1-\alpha)$ by optimality of the private-sum computation protocol, and because $H(K_{\mathcal{T}}) \geq n|\mathcal{T}| (1-\alpha)$ by Lemma \ref{lem4}.\end{enumerate} 
  \end{proof}

    We now prove Theorem \ref{th4} using Lemma \ref{lem6} as follows. For any $\mathcal{T}\subsetneq [L]$, we have 
\begin{align*}
  \frac{I(U;K_{\mathcal{T}})}{H(U)} &= \frac{H(U)-H(U|K_{\mathcal{T}})}{H(U)}\\
&= \frac{n (L-1) (1-\alpha )-n (L-|\mathcal{T}|-1) (1-\alpha )}{n (L-1) (1-\alpha )}\\
& =  \frac{|\mathcal{T}|}{L-1}, \numberthis\label{eqf1}
\end{align*}
where the second equality holds by Lemma \ref{lem6} and optimality of the private-sum computation protocol. Additionally, we have 
 \begin{align*}
   \frac{I(U;K_{[L]})}{H(U)} 
   &= \frac{H(U)-H(U|K_{[L]})}{H(U)}\\  &= 1, \numberthis\label{eqf2} \end{align*}
where the second equality holds because, by Lemma \ref{lem6},  $$H(U|K_{[L]})\leq H(U|K_{[L-1]})=0.$$
Hence, by \eqref{eqf1} and \eqref{eqf2}, we have 
\begin{align*}
  \frac{I(U;K_{\mathcal{T}})}{H(U)} = \begin{cases}
0  & \text{if } |\mathcal{T}|=0 \\
 \frac{|\mathcal{T}|  }{L-1}    & \text{if } |\mathcal{T}| \in  \llbracket 1, L-1  \rrbracket   \\
  1 & \text{if } |\mathcal{T}|=L
 \end{cases},
\end{align*}
which, by \eqref{eqrampsec}, means that $(K_l)_{l\in [L]}$ are the shares of a $(L-1, L-1 ,L)$ ramp secret sharing scheme where $U$ is the secret.

\section{Concluding remarks} \label{sec:cl}
We have studied the problem of private summation over a one-way, public, and noiseless channel from distributed users. Our setting generalizes the problem of secure summation by allowing a controlled amount of information leakage for a given set of colluding users. For a fixed amount of information leakage, we have characterized the optimal amount of required communication, global randomness, and local randomness at the users.
Under a leakage symmetry condition, which ensures that  all users are indistinguishable in terms of their capacity to gain information
about other users' private sequences, we have also derived the optimal individual rates of local randomness for each user. Additionally, we have established a connection between secret sharing and private summation by showing that secret sharing is necessary to generate the local randomness at the users.

\appendices

\section{Proof of \eqref{eqprivact2}} \label{app1}
For any $\mathcal{T} \subset [L]$, we have
\begin{align*}
&I(S_{[L]}; X_{[L]} |\Sigma_{[L]} ,S_{\mathcal{T}}, K_{\mathcal{T}})\\
& \leq H(S_{[L]} |\Sigma_{[L]} ,S_{\mathcal{T}} ,K_{\mathcal{T}}) \\
& \stackrel{(a)}\leq H(S_{[L]} |\Sigma_{[L]} ) \\
& \stackrel{(b)}= H(S_L | S_{[L-1]} ,\Sigma_{[L]}) + \sum_{l\in [L-1]} H(S_{l} |S_{[l-1]} ,\Sigma_{[L]} ) \\
& \stackrel{(c)}=  \sum_{l\in [L-1]} H(S_{l} |S_{[l-1]}, \Sigma_{[L]} ) \\
& \stackrel{(d)}\leq \sum_{l\in [L-1]} H(S_{l}  ) \\
& \stackrel{(e)}= n(L-1),
\end{align*}
where
\begin{enumerate}[(a)]
    \item holds because conditioning reduces entropy;
    \item holds by the chain rule;
    \item holds because $S_L$ can be recovered from $(S_{[L-1]} ,\Sigma_{[L]})$;
    \item holds because conditioning reduces entropy;
    \item holds by Definition \ref{def1}.
\end{enumerate}
Hence, the left-hand side of~\eqref{eqprivact2} is always upper bounded by $n(L-1)$.

\section{Proof of Properties \ref{prop1}--\ref{prop3}} \label{App_properties}
\subsection{Proof of Property \ref{prop1}}

    Let $l \in \llbracket 0,L-1\rrbracket$. Let $\mathcal{T} \subseteq [L]$ such that  $|\mathcal{T}|=l$ and $\mathcal{T} = (t_i)_{i\in [l]}$. Then, we have
    \begin{align*}
       l \Delta = \sum_{j\in [l]} [\mathbb{L}((t_i)_{i\in [j-1]}) - \mathbb{L}((t_i)_{i\in [j]}) ] = \mathbb{L}(\emptyset) - \mathbb{L} (\mathcal{T}).
    \end{align*}
    Hence, for any $\mathcal{T} \subseteq [L-1]$ such that  $|\mathcal{T}|=l$, we have $C_l = \mathbb{L}(\emptyset)-  l \Delta$.

\subsection{Proof of Property \ref{prop2}}

   Since $C_{L-1} =0$, we have
    \begin{align}
    (L-1) \Delta =    \sum_{l\in[L-1]} (C_{l-1} -C_l) = C_0 - C_{L-1} = C_0 = \mathbb{L}(\emptyset) \geq 0, \label{eqDeltavalue}
    \end{align}
    hence, for $l \in \llbracket 0,L-2\rrbracket$, $  0\leq \Delta = C_l - C_{l+1}$. We conclude that the property holds by noting that $C_{L-1} =0 = C_L$.

\subsection{Proof of Property \ref{prop3}}

    By \eqref{eqDeltavalue}, we have
    \begin{align*}
        \Delta = \frac{C_0}{L-1} \leq n\alpha .
    \end{align*}

\section{Proof of Lemma \ref{lem1}}\label{applem1}
For any $\mathcal{T} \subseteq [L]$, we have
\begin{align*}
    C_{|\mathcal{T}|} 
   & \stackrel{(a)}=I(S_{[L]}; X_{[L]} |\Sigma_{[L]} ,S_{\mathcal{T}} ,K_{\mathcal{T}}) \\
   & = I(S_{[L]}; X_{[L]} ,K_{\mathcal{T}}|\Sigma_{[L]} ,S_{\mathcal{T}} ) - I(S_{[L]};  K_{\mathcal{T}}|\Sigma_{[L]} ,S_{\mathcal{T}} )  \\
& \stackrel{(b)} = I(S_{[L]}; X_{[L]} ,K_{\mathcal{T}}|\Sigma_{[L]} ,S_{\mathcal{T}} )   \\
& = H(S_{[L]}|\Sigma_{[L]}, S_{\mathcal{T}} ) - H(S_{[L]}| X_{[L]} ,K_{\mathcal{T}} ,\Sigma_{[L]} ,S_{\mathcal{T}} )   \\
& \stackrel{(c)}= H(S_{\mathcal{T}^c}|\Sigma_{\mathcal{T}^c} ,S_{\mathcal{T}} ) - H(S_{\mathcal{T}^c}| X_{[L]} ,K_{\mathcal{T}}  , \Sigma_{[L]} ,S_{\mathcal{T}} )  \\
& \stackrel{(d)}= H(S_{\mathcal{T}^c}|\Sigma_{\mathcal{T}^c} ,S_{\mathcal{T}} ) - H(S_{\mathcal{T}^c}| X_{[L]} ,K_{\mathcal{T}}  , S_{\mathcal{T}} )  \\
&  \stackrel{(e)}= H(S_{\mathcal{T}^c}|\Sigma_{\mathcal{T}^c}  ) - H(S_{\mathcal{T}^c}| X_{[L]} ,K_{\mathcal{T}}   ,S_{\mathcal{T}} )  \\
& \stackrel{(f)}= H(S_{\mathcal{T}^c}) - H(\Sigma_{\mathcal{T}^c}  ) - H(S_{\mathcal{T}^c}| X_{[L]} ,K_{\mathcal{T}}   ,S_{\mathcal{T}} )  \\
& \stackrel{(g)}= n(L- |\mathcal{T}|) - n\mathds{1}\{ \mathcal{T}^c \! \neq \emptyset\} - H(S_{\mathcal{T}^c}| X_{[L]}, K_{\mathcal{T}}   ,S_{\mathcal{T}} ) ,
\end{align*}
where \begin{enumerate}[(a)]
    \item holds by Property \ref{prop1}; \item  holds because \begin{align*}
        I(S_{[L]};  K_{\mathcal{T}}|\Sigma_{[L]} ,S_{\mathcal{T}} )& \leq I(S_{[L]},\Sigma_{[L]} ,S_{\mathcal{T}};  K_{\mathcal{T}})\\
        &=0,\end{align*} where the equality holds by independence between the users's sequences and local randomness; \item holds by the chain rule and because there is a one-to-one mapping between $(\Sigma_{[L]} ,S_{\mathcal{T}})$ and $(\Sigma_{\mathcal{T}^c}, S_{\mathcal{T}})$;\item holds because $\Sigma_{[L]}$ can be reconstructed from $X_{[L]}$ by \eqref{eqrec}; \item holds by independence between $(S_{\mathcal{T}^c},\Sigma_{\mathcal{T}^c} )$ and $S_{\mathcal{T}}$; \item holds because $H(S_{\mathcal{T}^c},\Sigma_{\mathcal{T}^c}  ) = H(S_{\mathcal{T}^c})$;
            \item holds by uniformity of the users' sequences.
    \end{enumerate}

\bibliographystyle{IEEEtran}
\bibliography{bib}

\end{document}